\documentclass[11pt]{article}

\usepackage{preamble}

\newcommand{\eas}{\textsf{EAS}}
\newcommand{\weas}{\textsf{WEAS}}
\newcommand{\irange}[2]{\ensuremath{\lrb{#1, #2}}}

\title{Interweaving Real-Time Jobs with Energy Harvesting to \\ Maximize Throughput}
\author{Baruch Schieber\thanks{Department of Computer Science, New Jersey Institute of Technology. Emails: \email{\{sbar,bs567,sv96\}@njit.edu}}
\and Bhargav Samineni\footnotemark[1] 
\and Soroush Vahidi\footnotemark[1] 
}
\date{}

\begin{document}

\maketitle

\begin{abstract}
  Motivated by baterryless IoT devices, we consider the following scheduling problem. The input includes $n$ unit time jobs $\calJ = \lrc{J_1, \ldots, J_n}$, where each job $J_i$ has a release time $r_i$, due date $d_i$, energy requirement $e_i$, and weight $w_i$.
  We consider time to be slotted; hence, all time related job values refer to slots. Let $T=\max_i\lrc{d_i}$. The input also includes an $h_t$ value for every time slot $t$ $\lrp{1 \leq t \leq T}$, which is the energy harvestable on that slot.
  Energy is harvested at time slots when no job is executed.
  The objective is to find a feasible schedule that maximizes the weight of the scheduled jobs.
  A schedule is feasible if for every job $J_j$ in the schedule and its corresponding slot $t_j$, $t_{j} \neq t_{j'}$ if ${j} \neq {j'}$, $r_j \leq t_j \leq d_j$, and the available energy before $t_j$ is at least $e_j$.
  To the best of our knowledge, we are the first to consider the theoretical aspects of this problem.

  In this work we show the following. \textsf{(1)} A polynomial time algorithm when all jobs have identical $r_i, d_i$ and $w_i$.
  \textsf{(2)} A $\frac{1}{2}$-approximation algorithm when all jobs have identical $w_i$ but arbitrary $r_i$ and $d_i$. \textsf{(3)} An FPTAS when all jobs have identical $r_i$ and $d_i$ but arbitrary $w_i$.
  \textsf{(4)} Reductions showing that all the variants of the problem in which at least one of the attributes $r_i$, $d_i$, or $w_i$ are not identical for all jobs are $\NPH$.

\end{abstract}

\section{Introduction}
\label{sec:intro}
The energy aware scheduling problem considered in this paper is defined as follows. Its input includes $n$ jobs $\calJ = \lrc{J_1, \ldots, J_n}$, where each job $J_i$ is associated with release time $r_i$, due date $d_i$, energy requirement $e_i$, and weight $w_i$. All jobs have equal (unit) processing time.
We consider time to be slotted with unit length and all time related job values to refer to slots. Let $T=\max_i\lrc{d_i}$. The input also includes an $h_t$ value specified explicitly for every time slot $t$, $\lrp{1 \leq t \leq T}$, which is the energy harvestable on the slot. Energy can be harvested \emph{only} at time slots where no job is being executed.
The energy available immediately before slot $t$ is the energy harvested in slots $1, \ldots, t-1$ minus the energy consumed by the jobs scheduled in slots $1, \ldots, t-1$. A schedule is \emph{feasible}
if it schedules no more than one job at a time, all executed jobs are scheduled between their release time and due date, and the energy available immediately before a job is executed is at least its energy requirement.
The objective is to maximize throughput, namely, the weight of the jobs in a feasible schedule.

\subsection{Motivation}
Our problem is motivated by the proliferation of Internet of Things (IoT) devices, which are used for many applications such as sensor networks, control systems, and home and building automation, to name a few.
One of the major challenges impacting the deployment of these devices is their power source. Most are powered by batteries, which are compact and lightweight options. However, the chemicals contained in these batteries pose a considerable risk to our environment \cite{IN2020}.
Also, the scarcity of the materials needed for batteries makes them prohibitively expensive in some applications, especially on an industrial scale.
Battery maintenance is another major issue, as the limited lifetime of batteries requires expensive and constant care to replace or recharge them. Consider, for example, IoT devices with humidity and temperature sensors that are stationed along an oil pipeline in a remote area and connected to its SCADA system \cite{B2010}.
Or, consider IoT devices for wildlife monitoring that are attached to animals and used to gather information like migration paths and population mortality \cite{BSW2018}.
Connecting these devices to the electric grid is not an option in some places, and using batteries is a logistical nightmare.

To overcome these issues, batteryless IoT devices have been proposed \cite{SK2010,M2016}. The energy used by these devices is directly harvested from environmental and renewable sources such as solar, wind, and radio-frequency (RF).
In the simplest designs, the harvested energy output is connected directly to the load. However, this design is only appropriate when the harvested current and voltage match those required for a task's execution \cite{LBC2017}.
More sophisticated designs include either capacitors or super capacitors to store energy, and are \emph{intermittent} systems \cite{IN2020,LBC2017} in which energy harvesting (charging) and task execution (discharging) are mutually exclusive to allow for a single control thread.
This is captured in our model by interweaving job execution and energy harvesting.

A major challenge in the design of such an intermittent system is the variability of energy harvesting over time~\cite{SZ2016}. For example, consider a solar energy source.
Certainly, this energy can only be harvested in the daytime; but even during the day the amount of harvestable energy varies based on the cloud cover and the sun angle and its predictability is challenging~\cite{FAS2015}.
We consider the \emph{offline} version of the problem, as proposed in \cite{IN2020}, and assume that the energy harvesting profile over time is given as part of the input to our problem.

As noted in \cite{IN2020}, intermittent systems are primarily used in  monitoring and surveillance applications that collect data at a fixed rate and then process the data periodically.
We model this by associating a release time and due date for each job (task). Due to the bare-bones design of these systems, all these tasks are pretty basic and require a minimal number of cycles.
Thus, it can be assumed that all these tasks require equal processing time.

\subsection{Our results}
We present both algorithms and hardness results for several variants of the problem. To the best of our knowledge, our work is the first theoretical analysis of this problem.
Our objective is to find a feasible schedule that maximizes job throughput. When all the jobs have the same weight (i.e.\! the unweighted setting), this corresponds
to finding a schedule that maximizes the number of jobs scheduled. We call this problem \textsf{Energy Aware Scheduling (EAS)}. Otherwise, this corresponds to maximizing the weight of the jobs scheduled (i.e.\! the weighted setting).
We call this problem \textsf{Weighted Energy Aware Scheduling (WEAS)}.

In \cref{sec:optimalA} we give an optimal polynomial time algorithm for \eas{} when all the jobs have identical release times and due dates. The dynamic programming approach we use is based on some properties of an optimal solution for this case.
In \cref{sec:optimalB} we give a more efficient algorithm (that is more involved) for the same problem that is based on additional properties of an optimal solution.
In \cref{sec:greedy} we consider \eas{} when jobs have arbitrary release times and due dates and show that a simple greedy algorithm achieves a $\frac 12$-approximation.
Interestingly, the proof of the approximation ratio of this simple algorithm is quite ``tricky''.
In \cref{sec:weighted_fptas} we show an FPTAS for \weas{} when all the jobs have identical release times and due dates.
In \cref{subsec:nph_unweighted} we prove that \eas{} is weakly $\NPH$ whenever all jobs don't have both identical release times and identical due dates.
We also show that \weas{} is weakly $\NPH$ in \cref{subsec:nph_weighted}.

\subsection{Prior work}
Our model is inspired by a similar model proposed by Islam and Nirjon~\cite{IN2020}. While their model allows arbitrary processing times for both jobs and energy harvesting, we consider all processing times to be uniform.
Additionally, they model jobs as belonging to a set of periodic tasks, while we do not enforce such a constraint.
They give heuristic scheduling algorithms for the offline (in which the harvestable energy profile is part of the input) as well as the online (in which the harvestable energy profile is not known a priori) versions of the problem.
They benchmarked their algorithms against other heuristics and in the offline case also compared their algorithm to the optimal solution computed using an IP solver.

Our problem is related to scheduling with \emph{nonrenewable resources} in which jobs require a nonrenewable resource like energy or funding to be scheduled. In contrast to our problem, it is assumed that the replenishment of the resource is done instantly at predetermined times.
As in our problem, jobs can be scheduled feasibly only if the amount of available resource when they are started is at least their resource requirement.
The goal is to schedule these resource-consuming jobs to optimize various objectives. In \cite{GK2014} it was shown that computing a minimum makespan schedule of jobs with a single nonrenewable resource requirement with 2 replenishment times, arbitrary processing times, and identical release times on a single machine is (weakly) $\NPH$.
The same paper also proved that in case the number of replenishment times of the resource is part of the input, the problem is strongly $\NPH$.
\cite{GLW2011} considered the same setting but with the objective of delay minimization. \cite{GHD2005,CCE2018} (and references therein) also considered scheduling with nonrenewable resources.

Another related problem is \emph{inventory constrained scheduling}, which considers the scheduling of two types of jobs: resource-producing and resource-consuming jobs.
A resource-consuming job cannot be scheduled unless its resource requirement is available. Unlike our problem where there is no need to schedule all the resource-producing (energy harvesting) jobs, in this problem both the resource-producing and resource-consuming jobs have to be scheduled.
Several variants of this problem were considered in~\cite{BCL2008,BCL2010, DRD2020}. These variants include the case of jobs with equal processing times and the objective of minimizing the number of tardy jobs, which is the complement of unweighted throughput maximization.

\cite{B1999}, and later \cite{V2012}, considered the ``non-energy-aware'' version of our problem and showed that minimizing the number of tardy jobs, which is equivalent to maximizing the unweighted throughput, can be solved in polynomial time.
\cite{ADH2010} identified special cases in which minimizing the weighted tardiness of the ``non-energy-aware'' version of our problem is polynomial.

\subsection{Notations}
As we mostly deal with integral values, we use the notation $\irange{a}{b}$, where $a, b \in \Z^+$ and $a \le b$, to denote the set of integers between $a$ and $b$ inclusive.
We assume that time is slotted with unit length; that is, when referring to a job being scheduled or energy being harvested at time slot $j \in \Z^+$, it means they are done on the time interval $[j-1, j)$.
A job $J_i\in \mathcal{J} = \lrc{J_1, \ldots, J_n}$ can be executed at any slot in $\irange{r_i}{d_i}$, assuming that the energy available before its execution time slot is at least $e_i$.

We formally define a schedule $S$ as a pair $(J(S),\pi_S)$
where $J(S) \sse \mathcal{J}$ is the set of jobs scheduled by $S$ and $\pi_S: J(S) \rightarrow \irange{1}{T}$ is a mapping that maps each job $J_i \in J(S)$
to its execution time slot $t \in \irange{1}{T}$.
We additionally define $E_S(t)$ to be the energy amount available immediately before slot $t$ in schedule $S$, given by the equation
\begin{equation*}
  E_S(t) = \sum_{\tau=1}^{t-1} h_{\tau} - \lrp{\sum_{J_i \in \lrc{J^{\prime} \in J(S) \,|\, \pi_S(J^{\prime}) < t}} e_i + h_{\pi_S(J_i)}}.
\end{equation*}
That is, it is the total energy harvestable from all time slots $\tau \in \irange{1}{t-1}$
minus the energy consumed by the jobs scheduled up to slot $t$ and the energy harvestable at the slots when these jobs are scheduled (since energy can be harvested only at slots in which no job is scheduled).
A schedule $S$ is feasible if $\pi_S$ is one-to-one, and for each job $J_i \in \func{J}{S}$ scheduled on time slot $t_i = \func{\pi_S}{J_i}$, we have $\func{E_S}{t_i} \geq e_i$ and $t_i \in \irange{r_i}{d_i}$.

\section{An Optimal Algorithm for \eas{} when All Jobs Have Identical Release Times and Due Dates} \label{sec:optimalA}

In this section, we consider \eas{} instances in which all jobs have identical release times and due dates and present a polynomial time dynamic programming algorithm that produces an optimal schedule.
From now on, we assume that the jobs are sorted in non-decreasing order by their energy requirement ($e_1 \leq \ldots \leq e_n$).
We consider all $r_i = 1$, though the algorithm can easily be adapted to cases where all $r_i = r$ for some $r > 1$. 
We begin with two claims.

\begin{claim}
  There exists an optimal schedule that schedules a prefix of the sequence $J_1,\ldots,J_n$.
  \label{claim:prefix}
\end{claim}
\begin{proof}
  Suppose that the optimal algorithm schedules the $m$ jobs $J_{i_1},\ldots,J_{i_m}$, where $i_1<\cdots<i_m$.
  Then, the schedule that replaces job $J_{i_j}$ by $J_j$ for $j\in \irange{1}{m}$ is also feasible since $e_j \leq e_{i_j}$.
\end{proof}

\begin{claim}
  There exists an optimal schedule in which the jobs are scheduled in non-decreasing order of their energy requirement.
  \label{claim:nondec}
\end{claim}
\begin{proof}
  Consider a feasible schedule $S$ in which a job $J_i$ is scheduled before job $J_j$ and $e_i > e_j$.
  To prove the claim it is sufficient to show that the schedule $S^{\prime}$ given by swapping $J_i$ and $J_j$ is also feasible.
  Let $t_i=\pi_S(J_i)$ and $t_j=\pi_S(J_j)$.
  The parts of schedule $S^{\prime}$ immediately before time slot $t_i$ and after time slot $t_j$ are feasible since for every $t \in \irange{1}{t_i-1} \cup \irange{t_j+1}{T}$, we have $E_{S^{\prime}}(t) = E_S(t)$.
  Job $J_j$ can be scheduled at slot $t_i$ since $E_{S^{\prime}}(t_i) = E_S(t_i) \geq e_i > e_j$. For every slot $t \in \irange{t_i + 1}{t_j-1}$, we have  $E_{S'}(t)=E_S(t)+e_i-e_j > E_S(t)$.
  Thus, the jobs scheduled in $S'$ in these time slots are also feasible.
  Now consider slot $t_j$. Job $J_j$ is scheduled in $S$ at this slot, hence $E_S(t_j) \geq e_j$.
  It follows that $E_{S'}(t_j) = E_S(t_j) + e_i - e_j \geq e_i$, which implies $J_i$ can be scheduled at slot $t_j$. Therefore, $S^{\prime}$ is a feasible schedule.
  Since the optimal schedule must also be feasible, the same swapping procedure can be applied.
\end{proof}

We apply these observations to obtain a dynamic programming algorithm to compute an optimal schedule $O$. Define the dynamic programming ``table'' as follows.
For $i \in \irange{1}{n}$ and $t \in \irange{1}{T}$, let $A(i,t)$ be the maximum amount of available energy at the start of time slot $t+1$, where the maximum is taken over all feasible schedules of jobs $J_1,\ldots,J_i$ on the time slots $\irange{1}{t}$. 
If such a feasible schedule does not exist, then $A(i,t)=-\infty$.
Since the input size is $\Omega(n+T)$, the size of this table is polynomial. The maximum number of jobs that can be scheduled feasibly is given by the maximum $m$ for which $A(m,T) \geq 0$.
The respective optimal schedule can be computed by backtracking the intermediate values that contributed to $A(m,T)$.
The computation of $A(i,t)$ is given in \cref{alg:dynamic}, whose time complexity is $\bigo{nT}$. In \cref{sec:optimalB} we describe a more efficient algorithm with time complexity $\bigo{n\log n+T}$ that is also optimal.

\begin{algorithm}[!h]
	\caption{} \label{alg:dynamic}
	\begin{algorithmic}[1]
    \Algin (1) $n$ jobs $\{J_1,\ldots,J_n\}$, each with $r_i=1$, $d_i=T$, and energy requirement $e_i$, and (2) $h_t$ for each time slot $t \in \irange{1}{T}$
		\Algout $A(i,t)$, for $i \in \irange{1}{n}$ and $t \in \irange{1}{T}$
    \State $A(\cdot,\cdot) \gets -\infty$ \Comment{Initialize the table}
    \For {$t = 2$ to $T$}
      \If {$\sum_{j=1}^{t-1} h(j) \ge e_1$}
        \State $A(1,t) \gets \max\{A(1,t-1)+h(t),\sum_{j=1}^{t-1} h(j) - e_1\}$
      \EndIf
    \EndFor
    \For {$i = 2$ to $n$}
      \For {$t = i+1$ to $T$}
        \If {$A(i-1,t-1) \ge e_i$}
          \State $A(i,t) \gets \max\{A(i,t-1)+h(t),A(i-1,t-1) -e_i\}$
        \EndIf
      \EndFor
    \EndFor
	\end{algorithmic}
\end{algorithm}

\begin{theorem}
  The maximum  number of jobs that can be scheduled feasibly is given by the maximum $m$ for which $A(m,T) \geq 0$ in the array $A(\cdot,\cdot)$ computed by \cref{alg:dynamic}.
\end{theorem}
\begin{proof}
  Let $m$ be the maximum number of jobs that can be scheduled feasibly. For $i \in \irange{1}{m}$, let $t_i$ be the time slot in which $J_i$ is scheduled in such a schedule.
  It is easy to see that in this case $A(i,t_i) \geq 0$, for every $i \in \irange{1}{m}$. In the other direction, suppose that $A(i,t)\geq 0$. In this case, there exists a feasible schedule of jobs $J_1,\ldots,J_i$
  on the time slots $\irange{1}{t}$. The schedule can be computed by backtracking the intermediate values that contributed to $A(i,t)$.
\end{proof}

\section{A More Efficient Optimal Algorithm for \eas{} when All Jobs Have Identical Release Times and Due Dates}
\label{sec:optimalB}

We again assume that the jobs are sorted in non-decreasing order by the energy requirement ($e_1 \leq \ldots \leq e_n$).
We also assume that all values of $e_i$, for $i \in \irange{1}{n}$ and all values of $h_t$, for $t \in \irange{1}{T}$, are distinct.
This can be assumed without loss of generality since we can always break ties based on the index $i$ in the first case and on $t$ in the latter.
We also make use of \cref{claim:prefix,claim:nondec}.
Namely, we only consider schedules that schedule a prefix of the job sequence $J_1,\ldots, J_n$ in sorted order.

\begin{lemma}\label{lem:unique}
  Among all (considered) schedules that schedule a fixed number of jobs $k \in \irange{1}{n}$ (if such exist), there is a unique schedule $S_k$ 
  that maximizes the available energy at the start of time slot $T+1$ over all such schedules.
\end{lemma}
\begin{proof}
  Suppose for the sake of contradiction that there are two such schedules, $S$ and $S^{\prime}$, with differing mapping functions that maximize the available energy at time slot $T+1$.
  Let $t$ be the earliest time slot at which the two schedules differ. Without loss of generality assume that $S$ schedules a job on time slot $t$ while $S'$ schedules the same job on a later time slot.
  Since both $S$ and $S'$ schedule the same number of jobs, there must be at least one time slot 
  later than $t$
  when $S^{\prime}$ schedules a job while $S$ does not. Let $s'>t$ be the earliest such time slot.
  Let $s \in \irange{t}{s^{\prime}-1}$ be the latest time slot (before slot $s'$) in which $S$ schedules a job while $S'$ does not. It follows that for any time slot $t^{\prime} \in \irange{s+1}{s'-1}$, $S$ schedules a job at slot $t^{\prime}$ if and only if $S'$ schedules a job at the same slot.
  Let $t_1,\ldots, t_\ell$ be all the time slots in  $\irange{s+1}{s'-1}$ that are used by $S$ (and thus also by $S'$) to schedule jobs, if such time slots exist. Let $t_{\ell+1}=s'$.
  Suppose that $S$ schedules job $J_i$ at time slot $s$. It follows that $J_i$ is scheduled in $S'$ later than $s$. Now, compare $h_s$ and $h_{s^\prime}$. If $h_{s^\prime} < h_s$ then the schedule obtained from $S$ by moving job $J_i$ to time slot $t_1$, and moving the job scheduled at $t_i$ to $t_{i+1}$, for $i\in \irange{1}{\ell}$,
  is also feasible and results in more available energy at slot $T+1$ which is a contradiction.
  If $h_{s^\prime} > h_s$, then
  the schedule obtained from $S'$ by moving the job scheduled at $t_1$ to $s$, and moving the job scheduled at $t_{i+1}$ to $t_{i}$, for $i\in \irange{1}{\ell}$,
  is feasible since if a job $J_j$ is moved to slot $t$ then we are guaranteed that schedule $S$ scheduled a job $J_{j'}$, for some $j' \geq j$, at $t$. We again reach a contradiction since the revised schedule results in more available energy at slot $T+1$. The proof follows.
\end{proof}

For any $i\in \irange{1}{m}$ such that $J_1,\ldots,J_i$ can be scheduled feasibly, let $S_i$ be the unique schedule that schedules $J_1,\ldots,J_i$ in that order and maximizes the available energy at time $T+1$ among all such schedules.

\begin{lemma}\label{lem:sametimes}
  For any $i>1$, if $S_i$ is defined, then $i-1$ of the time slots used by schedule $S_i$ are the time slots used by $S_{i-1}$.
\end{lemma}
\begin{proof}
  To obtain a contradiction suppose that $S_i$ does not use all the time slots used by $S_{i-1}$. Let $t$ be the earliest time slot at which schedule $S_{i-1}$ schedules a job while schedule $S_i$ does not.
  If there exists a slot $t''<t$ at which schedule $S_i$ schedules a job while schedule $S_{i-1}$ does not, then let $t'<t$ be the latest such slot. We reach a contradiction as in the proof of \cref{lem:unique}, where $t'$ and $t$ play the role of $s$ and $s'$, respectively.
  Otherwise, there must be a time slot $s''>t$ when schedule $S_i$ schedules a job from $\lrc{J_1,\ldots,J_{i-1}}$ while schedule $S_{i-1}$ does not. Let $s'$ be the earliest such time slot. Let $s \in \irange{t}{s'-1}$ be the latest time slot (before slot $s'$) in which schedule $S_{i-1}$ schedules a job while schedule $S_i$ does not.
  Again, a contradiction is reached as in the proof of \cref{lem:unique}.
\end{proof}

\begin{algorithm}[ht]
	\caption{}\label{alg:unique}
	\begin{algorithmic}[1]
		\Algin (1) $n$ jobs $\{J_1,\ldots,J_n\}$, each with release time $r_i=1$, due date $d_i=T$, and energy requirement $e_i$, and (2) $h_t$ for each time slot $t \in \irange{1}{T}$
		\Algout the set of time slots used by the optimal schedule
    \State $t_0 \gets 0$
    \For {$i = 1$ to $n$}
      \State $E_{\min} \gets \infty$
      \State $t_i \gets T+1$
      \For {$j = 1$ to $i$}
        \If {$\exists t\in (t_{j-1},t_j)$ s.t. the schedule that uses time slots $\{t_1,\ldots,t_{i-1}\}\cup\{t\}$ is feasible}
          \State Let $s$ be such a time slot with the minimum harvestable energy
          \If {$h_s < E_{\min}$}
            \State $E_{\min} \gets h_s$
            \State $s_{\min} \gets s$
            \State $j_{\min} \gets j$

          \EndIf
        \EndIf
      \EndFor
      \If {$E_{\min} = \infty$}
        \State \Return $\{t_1,\ldots,t_{i-1}\}$
      \Else
        \State $\{t_1,\ldots,t_i\} \gets \{t_1,\ldots,t_{j_{\min}-1}, s, t_{j_{\min}},\ldots,t_{i-1}\}$
      \EndIf
    \EndFor
	\end{algorithmic}
\end{algorithm}

\cref{lem:sametimes} implies \cref{alg:unique}. We note that this algorithm can be implemented efficiently by using a priority queue to find the interval $(t_j,t_{j+1})$ to which the additional time slot used by $S_i$ belongs. The time complexity of the resulting algorithm is $\bigo{n\log n+T}$. 
We conclude with the following theorem whose proof is a direct consequence of \cref{lem:unique,lem:sametimes}.
\begin{theorem}
  \cref{alg:unique} returns the set of time slots that can be used to feasibly schedule the maximum  number of jobs.
\end{theorem}

\section{A Greedy $\frac{1}{2}$-Approximation for \eas{}}
\label{sec:greedy}
This section considers the general case of \eas{} when jobs have arbitrary release times and due dates. This problem variant is $\NPH$ as shown in \cref{subsec:nph_unweighted}.
We present a $\frac{1}{2}$-approximation for this case that uses a greedy scheduling strategy.

Consider the following greedy approach to scheduling jobs.
The algorithm works in iterations where in each iteration, either one job is scheduled or the algorithm stops.
Let $U$ be the set of unscheduled jobs and $G$ the schedule constructed by the algorithm. Initially, $U=\lrc{J_1,\ldots,J_n}$ and $J(G) = \emptyset$.
In iteration $\ell$, the algorithm first checks for each job $J_i \in U$ whether there exists at least one time slot it can be feasibly scheduled in without impacting the feasibility of previously scheduled jobs.
If it is not feasible to schedule any of the jobs, then the algorithm stops. Otherwise, for each job $J_i$ that can be scheduled feasibly, find the time slot $t_i$ that minimizes $Q = e_i+h_{t_i}$ over all its feasible time slots.
The job scheduled in iteration $\ell$ is the job that minimizes $Q$ over all feasible jobs that can be scheduled during this iteration. Denote this job as $J_j$, remove it from $U$, and add it to $J(G)$ with $\pi_G(J_j) = t_j$.
The algorithm's pseudocode is given in \cref{alg:general}.

\begin{algorithm}[!h]
	\caption{} \label{alg:general}
	\begin{algorithmic}[1]
    \Statex \emph{Input}: (1) $n$ jobs $\{J_1,\ldots,J_n\}$, each with release time $r_i$, due date $d_i$, and energy requirement $e_i$, and (2) $h_t$ for each time slot $t \in \irange{1}{T}$
    \Statex \emph{Output}: A feasible schedule $G$
    \State $U \gets \lrc{J_1,\ldots,J_n}$
    \State $J(G) \gets \emptyset, \func{\pi_G}{\cdot} \gets \emptyset$
    \State $G \gets \lrp{J(G), \pi_G}$
    \State $Q_{\min} \gets 1$
    \Comment{The tentative minimum energy value}
    \While {$U \neq \emptyset$ \textbf{and} $Q_{\min}>0$}
      \State $E \gets 0$ \Comment{Tracks the energy available}
      \State $Q_{\min} \gets 0$
      \For{$t = 1$ to $T$}
        \If {a job $J_i$ is already scheduled at slot $t$}
          \State $E \gets E - e_i$
          \LComment{Tentatively scheduled job causes already scheduled jobs to be infeasible}
          \If {$E < 0$}
            \State $E \gets E + Q_{\min}$
            \State $Q_{\min} \gets 0$
          \EndIf
        \ElsIf{$\lrc{ J_i\in U \,|\, t\in \irange{r_i}{d_i}}= \emptyset$}
          \State  $E \gets E + h_t$
        \Else
          \State Let $J_k$ be the job that minimizes $e_k$ over jobs in $\lrc{J_i\in U \,|\, t \in \irange{r_i}{d_i}}$
          \If{$E + Q_{\min} \geq e_k$ \textbf{and} $\left(Q_{\min}=0  \textbf{ or }  e_k + h_t < Q_{\min} \right)$}
            \State $E \gets E + Q_{\min} - e_k$
            \Comment{$J_k$ becomes the tentatively scheduled job}
            \State $Q_{\min} \gets e_k + h_t$
            \State $j \gets k$
            \State $t_j \gets t$
          \Else
            \State  $E \gets E + h_t$
          \EndIf
        \EndIf
      \EndFor
      \LComment{Schedule the job with index $j$ at slot $t_j$ if it is feasible}
      \If {$Q_{\min}>0$}
        \State $J(G) \gets J(G) \cup \lrc{J_j}$
        \State $\pi_G(J_j) \gets t_j$
        \State $U \gets U \sm \lrc{J_j}$
      \EndIf
    \EndWhile
	\end{algorithmic}
\end{algorithm}

Let $O=(J(O),\pi_O)$ be an optimal schedule, where $\abs{J(O)} = m \leq n$. Suppose that $\abs{J(G)} = x$, which implies that the greedy algorithm stops after completing $x$ iterations.
Let $J^g_1,\ldots,J^g_x$ be the jobs scheduled by the greedy algorithm, where $J^g_\ell$ is scheduled in iteration $\ell$. We prove the following lemma, which will later be used
to prove the approximation ratio.

\begin{lemma}
  At the end of iteration $\ell$ of the greedy algorithm, for $1\le \ell \le x$, there exists
  a feasible schedule $S$ such that \textbf{(i)} $\abs{J(S)} \ge \max\lrc{\ell,m-\ell}$, \textbf{(ii)} $\lrc{J_1^g, \ldots, J_{\ell}^g} \sse J(S)$ with $\pi_S\lrp{J_i^g} = \pi_G\lrp{J_i^g}$, for ${1}\le i\le {\ell}$,
  and \textbf{(iii)} $J(S) \sm \lrc{J_1^g, \ldots, J_{\ell}^g} \sse J(O)$. 
  \label{lem:greedy}
\end{lemma}
\begin{proof}
  We prove the lemma by induction. For the induction base we add a ``dummy'' iteration $0$ before the actual start of the greedy algorithm. The claim holds for $\ell=0$  since at the beginning of the greedy algorithm the schedule $O$ is feasible and $|J(O)|=m$.
  Consider the end of iteration $\ell$, for $\ell \geq 1$. By the inductive hypothesis, at the start of iteration $\ell$, there exists a feasible schedule $S$ of at least $\max\lrc{\ell-1,m-\ell+1}$ jobs that schedules the jobs $J^g_1,\ldots,J^g_{\ell-1}$ at time slots $\pi_G(J^g_1),\ldots,\pi_G(J^g_{\ell-1})$, respectively, and the remaining jobs belong to $J(O)$.
  Suppose that $S$ schedules job $J^g_\ell$ on time slot $t_{\ell} = \pi_G(J^g_\ell)$. In this case $S$ satisfies the conditions of the lemma also for $\ell$, since it schedules at least $\max\lrc{\ell,m-\ell}$ jobs, including the jobs $J^g_1,\ldots,J^g_{\ell}$ at slots $\pi_G(J^g_1),\ldots,\pi_G(J^g_{\ell})$, and the remaining jobs belong to $J(O)$.

  From now on assume that $S$ has not scheduled job $J^g_\ell$ at time slot $t_{\ell}$. We show how to obtain a schedule $S^{\prime}$ that satisfies the conditions of the lemma for $\ell$.
  We start with schedule $S$ and modify it as follows. First, we schedule the job $J^g_{\ell}$ at slot $t_{\ell}$ (in case $J^g_\ell \in J(S)$ this would just change the execution time slot of $J^g_\ell$).
  If $S$ already scheduled another job at slot $t_{\ell}$, then this job is discarded. Otherwise (that is, if $S$ has not scheduled another job at slot $t_{\ell}$), then the job in $J(S) \sm \lrc{J^g_1,\ldots,J^g_{\ell}}$ that $S$ schedules earliest is discarded.

  Before showing that schedule $S^{\prime}$ is feasible, we show that it satisfies the conditions of the lemma. Clearly, $S^{\prime}$ schedules the jobs $J^g_1,\ldots,J^g_{\ell}$ at time slots $\pi_G(J^g_1),\ldots,\pi_G(J^g_{\ell})$, and the remaining jobs in $J(S')$ belong to $J(O)$.
  Also, since exactly one job is discarded from $J(S)$, $\abs{J(S')} \geq \max\lrc{\ell,m-\ell}$. 

  Schedule $S'$ is feasible if and only if for all $J_i \in J(S'), \, E_{S'}(\pi_{S'}(J_i)) \geq e_i$. This is clearly the case for all jobs in $J(S')$ that are scheduled at slots $\irange{1}{t_{\ell}-1}$ since schedule $S$ is feasible.
  Next, consider job $J^g_{\ell}$ scheduled in $S'$ at slot $t_{\ell}$ and the rest of the jobs in $S'$ that are scheduled after this time slot. We distinguish between three cases.

  \par\noindent\textbf{Case 1:} $S$ schedules another job $J_j\in J(S)\sm\{J^g_1,\ldots,J^g_{\ell}\}$ at time slot $t_{\ell}$. In this case $E_{S'}(t_{\ell})=E_{S}(t_{\ell}) \geq e_j$.
  However, since the greedy algorithm preferred to schedule job $J_i=J^g_{\ell}$ at slot $t_{\ell}$ while $J_j$ was also feasible at the same time, we must have $e_i \leq e_j$ and thus $J_i = J^g_{\ell}$ is feasible at slot $t_{\ell}$.
  Since $e_i+h_{t_{\ell}} \leq e_j+h_{t_{\ell}}$, we have that for all $t \in \irange{t_{\ell}+1}{T}, \, E_{S'}(t) \geq E_{S}(t)$, and thus the jobs in $S'$ scheduled after time slot $t_{\ell}$ are also feasible.

  In the remaining two cases, $S$ has not scheduled another job at slot $t_{\ell}$. In these cases, the job in $J(S)\sm\{J^g_1,\ldots,J^g_{\ell}\}$ that $S$ schedules earliest is discarded. Let this job be denoted by $J_j$ and the
  time slot it was scheduled in by $t_j = \pi_S(J_j)$.

  \par\noindent\textbf{Case 2:} $t_j < t_{\ell}$. Again, since the greedy algorithm preferred to schedule $J_i=J^g_{\ell}$ at time slot $t_{\ell}$ while $J_j$ was also feasible at time slot $t_j$, we must have $e_i + h_{t_{\ell}} \leq e_j + h_{t_j}$.
  This implies that $E_{S'}(t_{\ell}) \geq E_{S}(t_{\ell}) + e_j + h_{t_j} \geq e_i$ and thus $J_i=J^g_{\ell}$ is feasible at slot $t_{\ell}$. It also implies that for all $t \in \irange{t_{\ell}+1}{T}, \, E_{S'}(t) = E_{S}(t) + e_j + h_{t_j} - e_i - h_{t_{\ell}} \ge E_{S}(t)$, and thus the jobs in $S'$ scheduled after time slot $t_{\ell}$ are also feasible.

  \par\noindent\textbf{Case 3:} $t_j > t_{\ell}$. Again, since the greedy algorithm preferred to schedule $J_i=J^g_{\ell}$ at time slot $t_{\ell}$ while $J_j$ was also feasible at time slot $t_j$, we must have $e_i + h_{t_{\ell}} \leq e_j+h_{t_j}$.
  This implies that for all $t \in \irange{t_j+1}{T}, \, E_{S'}(t)\ge E_{S}(t)$, and thus the jobs in $S'$ scheduled after time slot $t_j$ are feasible. Since $J_j$ is the earliest job in $J(S) \sm \lrc{J^g_1,\ldots,J^g_{\ell}}$, all the jobs in $J(S')$ scheduled before $t_j$ are in $\lrc{J^g_1,\ldots,J^g_{\ell}}$.
  Since the greedy schedule is guaranteed to be feasible, the schedule of these jobs in $S'$ is also feasible.

  Therefore, the schedule $S^{\prime}$ generated from modifying $S$ is always feasible. As $S^{\prime}$ was already shown to satisfy the constraints of the lemma, we have the proof of the inductive step.
\end{proof}

\begin{theorem}
  The greedy algorithm yields a $\frac{1}{2}$-approximation of the optimal solution.
\end{theorem}
\begin{proof}
  Consider any $1 \leq \ell \leq \ceil{\frac{m}{2}}$. \cref{lem:greedy} implies that at the start of iteration $\ell$ there exists a feasible schedule of at least $m-(\ell-1) \geq \ell$ jobs that schedules the jobs $J^g_1,\ldots,J^g_{\ell-1}$ at times $\pi_G(J^g_1),\ldots,\pi_G(J^g_{\ell-1})$.
  Thus, there exists at least one job that can be feasibly scheduled in iteration $\ell$. It follows that
  the greedy algorithm completes at least $\ceil{\frac{m}{2}}$ iterations, which implies that $2x \ge m$.
\end{proof}

\section{An FPTAS for \weas{} when All Jobs Have Identical Release Times and Due Dates}
\label{sec:weighted_fptas}
We now consider \weas{} in the special case of jobs with identical release times and due dates. This variant of the problem is $\NPH$ as shown in \cref{subsec:nph_weighted}.
We present a fully polynomial time approximation scheme (FPTAS) that for any constant $\varepsilon$ finds a feasible schedule that is a $\lrp{1-\varepsilon}$-approximation to the maximum weight of the scheduled jobs in any feasible schedule.
From now on, we assume that the jobs are sorted in non-decreasing order by their energy requirement ($e_1 \leq \ldots \leq e_n$).
We consider all $r_i = 1$, though the algorithm can easily be adapted to cases where all $r_i = r$ for some $r > 1$. 
We also assume WLOG that for any job $J_i \in \calJ$, there is a feasible schedule consisting of only $J_i$ (otherwise, we can omit $J_i$ from the input).
Let $W_{\max}= \max_{i} \lrc{w_i}$, where all the weights are assumed to be integers. Note that the optimal weight is in the interval $[W_{\max},nW_{\max}]$.

It is easy to see that \cref{claim:nondec} also holds for the weighted setting with the same proof as before. Thus, we can consider only feasible schedules in which the jobs are scheduled in non-descending order of their energy consumption.
Similar to \cref{sec:optimalA}, we first develop a dynamic programming algorithm to compute an optimal schedule $O$. However, the size of the dynamic programming ``table'' in this case is pseudo polynomial.
We show how to make the size of the table polynomial at the expense of losing optimality and replacing it by a $\lrp{1-\varepsilon}$ approximation.

For $i \in \irange{1}{n}$, $t\in \irange{1}{T}$, and $w \in \irange{0}{nW_{\max}}$, let $A(i,t,w)$ be the maximum amount of available energy at the start of time slot $t+1$, where the maximum is taken over all feasible schedules of subsets of $\lrc{J_1,\ldots,J_i}$ on the time slots $\irange{1}{t}$ such that the weight of the scheduled jobs is at least $w$.
If such a feasible schedule does not exist, then $A(i,t,w)=-\infty$.
The maximum weight of jobs that can be scheduled feasibly is given by the maximum $w$ for which $A(n,T,w) \geq 0$. The respective optimal schedule can be computed by backtracking the intermediate values that contributed to $A(n,T,w)$.
The computation of $A\lrp{i, t, w}$ is given in \cref{alg:weighted}.

\begin{algorithm}[!h]
	\caption{} \label{alg:weighted}
	\begin{algorithmic}[1]
    \Algin (1) $n$ jobs $\{J_1,\ldots,J_n\}$, each with $r_i=1$, $d_i=T$, energy requirement $e_i$, and weight $w_i$, and (2) $h_t$ for each time slot $t \in \irange{1}{T}$
		\Algout $A(i,t,w)$, for $i \in \irange{1}{n}$, $t\in \irange{1}{T}$, and $w \in \irange{0}{nW_{\max}}$
    \State $A(\cdot,\cdot,\cdot) \gets -\infty$ \Comment{Initialize the table}
    \For {$i = 1$ to $n$}
      \For {$t = 1$ to $T$}
        \State $A(i,t,0) \gets \sum_{j=1}^{t} h(j)$
      \EndFor
    \EndFor
    \For {$t = 2$ to $T$}
      \If {$\sum_{j=1}^{t-1} h(j) \ge e_1$}
        \For {$w = 1$ to $w_1$}
          \State $A(1,t,w) \gets \max \lrc{A(1, t-1, w) + h(t), \sum_{j=1}^{t-1} h(j)-e_1}$
        \EndFor
      \EndIf
    \EndFor
    \For {$i = 2$ to $n$}
      \For {$t = 2$ to $T$}
        \For {$w = 1$ to $iW_{\max}$}
          \If {$A(i-1,t-1,\max\lrc{0,w-w_i}) \ge e_i$}
            \State $A(i,t,w) \gets \max \{A(i-1, t, w), A(i,t-1,w)+h(t),A(i-1,t-1, \max\lrc{0,w-w_i}) -e_i\}$
          \Else
            \State $A(i,t,w) \gets \max \lrc{A(i-1, t, w), A(i, t-1, w) + h(t)}$
          \EndIf
        \EndFor
      \EndFor
    \EndFor
	\end{algorithmic}
\end{algorithm}

\begin{theorem}
  The maximum  weight of jobs that can be scheduled feasibly is given by the maximum $w$ for which $A(n,T,w) \geq 0$ in the array $A(\cdot,\cdot,\cdot)$ computed by \cref{alg:weighted}.
  \label{thm:weighted}
\end{theorem}
\begin{proof}
  Let $O$ be the optimal schedule, where $J(O) = \lrc{J_{i_1},\ldots,J_{i_m}}$. For $j\in \irange{1}{m}$, let $t_j=\pi_O(J_{i_j})$ be the time in which $J_{i_j}$ is scheduled, and let $W_j=\sum_{\ell=1}^j w_{i_\ell}$.
  It is easy to see that in this case $A(i_j,t_j,W_j)\ge 0$. In the other direction, suppose that $A(i,t,w)\ge 0$. In this case there exists a feasible schedule of a subset of the jobs $\lrc{J_1,\ldots,J_i}$ on the time slots $\irange{1}{t}$.
\end{proof}

We convert the pseudo polynomial algorithm to an FPTAS by rounding down the weights. Fix an $\varepsilon>0$. Suppose that all weights are rounded down to the nearest multiple of $\varepsilon W_{\max}/n$.
Since the weights are rounded down, every solution of the rounded down instance corresponds to a solution of the original instance with at least the same weight.
On the other hand, an optimal solution of the original instance corresponds to a solution of the rounded down instance whose weight is at least $\lrp{1-\varepsilon}$ times the weight of the optimal solution.
This is because the total decrease in weight due to rounding is no more than $\varepsilon W_{\max}$, and since the optimal weight is at least $W_{\max}$, this loss is no more than $\varepsilon$ fraction of the optimal solution.
It follows that the optimal solution of the rounded down instance corresponds to a solution of the original instance whose weight is at least $\lrp{1-\varepsilon}$ times the optimal weight of the original instance.
We compute the optimal solution using a dynamic programming approach as shown above with only one difference: the dynamic programming ``table'' is defined for coordinates $A(i,t,x)$, where $i \in \irange{1}{n}$, $t\in \irange{1}{T}$, and $x$ is an integer that represents the rounded weight $x \varepsilon W_{\max}/n$.
Since the optimal weight is bounded by $nW_{\max}$, $x$ is in the range $\irange{0}{n^2/\varepsilon}$, which is polynomial in $n$ and $1/\varepsilon$. We conclude with the following theorem.
\begin{theorem}
  There exists an FPTAS for \weas{} when all jobs have identical release times and due dates.
  \label{thm:fptas}
\end{theorem}

\section{Hardness Results}
\label{sec:nph}
In this section, we show the hardness of different variants of the problem in both the
unweighted and weighted settings.

\subsection{Unweighted setting}
\label{subsec:nph_unweighted}
In \eas{} we consider the cases of
arbitrary release times and identical due dates, identical release times and arbitrary due dates, and arbitrary release times and arbitrary due dates
and show they are all (weakly) $\NPH$. We mainly use a reduction from the $\ksum$ problem
(i.e.\! the parameterized version of \subsetsum{} whose hardness is shown in \cite{PW2010}).

\begin{theorem}
    \label{thm:nph_arb_release}
    \eas{} when jobs have arbitrary release times and identical due dates is (weakly) $\NPH$.
\end{theorem}
\begin{proof}
    A reduction is given from $\ksum$. An instance of this problem consists of a set of positive integers $\mathcal{A} = \lrc{\alpha_1, \ldots, \alpha_{n}}$, a target value $0<\beta < S = \sum_{i = 1}^{n} \alpha_i$, and an integer $k < n$.
    The objective is to decide whether there exists a subset $A \subseteq \mathcal{A}$ such that $\abs{A} = k$ and $\sum_{\alpha \in A} \alpha = \beta$.
    To simplify notation, we assume that $\mathcal{A}$ is sorted in non-increasing order ($\alpha_1 \geq \ldots \geq \alpha_{n}$)
    and assume without loss of generality that $S>2$ and $n>2$. Construct a corresponding instance of \eas{} with the following:
    \begin{itemize}
        \item
        Time slots $\irange{1}{2n-k+2}$ and a threshold value of $n$ jobs
        \item
        $\mathcal{J} = \lrc{J_1, \ldots, J_{n}}$ where each job $J_i \in \mathcal{J}$ has release time, due date, and energy requirement given by $r_i = i+1, d_i = 2n - k + 2,$ and $e_i = S^2 n^2 + \alpha_i\lrp{S n}$, respectively
        \item An energy harvesting profile defined by
        \[
        h_t = \begin{cases}
                    k \lrp{S^2 n^2} + \beta\lrp{S n} & \text{if } t = 1\\
                    S - \alpha_{t-1} & \text{if } t\in \irange{2}{n+1}\\
                    (n-k)(S^2 n^2) + (S-\beta)\lrp{S n} - S(n-k-1) - \beta & \text{if } t = n + 2\\
                    0 & \text{otherwise}
                \end{cases}
        \]
    \end{itemize}

    We claim that any feasible schedule that schedules all $n$ jobs must follow the structure shown in \cref{fig:ksum_reduction_arb_rel}.
    Notably, such a schedule must assign exactly $k$ jobs to time slots in $\tau_1 \coloneqq  \irange{2}{n+1}$ and the remaining $n-k$ jobs to time slots in $\tau_2 \coloneqq \irange{n+3}{2n - k +2}$.
    To see this, note that energy must always be harvested at the first time slot since no jobs are yet released.
    Moreover, the total amount of energy that can be harvested on $\tau_1$ is $\sum_{t=2}^{n+1} \lrp{S - \alpha_{t-1}} = S(n-1)$.
    Thus, the total amount of energy that can be harvested on slots $\irange{1}{n+1}$ is
    \begin{equation*}
        k\lrp{S^2 n^2} + \beta\lrp{Sn} + S(n-1) < k \lrp{S^2 n^2} + Sn(S + 1) < (k+1)(S^2 n^2),
    \end{equation*}
    which is strictly less than the energy required to execute more than $k$ jobs.
    This implies that no more than $k$ jobs can be scheduled in $\tau_1$, and that energy must be harvested at time slot $n+2$ in order to schedule more jobs.
    This leaves exactly $n-k$ time slots in $\tau_2$ until the due date to schedule the remaining $n-k$ jobs.
    Additionally, we note that by the construction of the energy harvesting profile, it is optimal to schedule jobs in $\tau_1$ immediately at their release time as the later a job is scheduled, the more energy is lost from not harvesting energy at that time slot.
    We proceed with the assumption that jobs in $\tau_1$ are scheduled in this way.

    \begin{figure}[ht]
        \centering
        \includegraphics[scale=1]{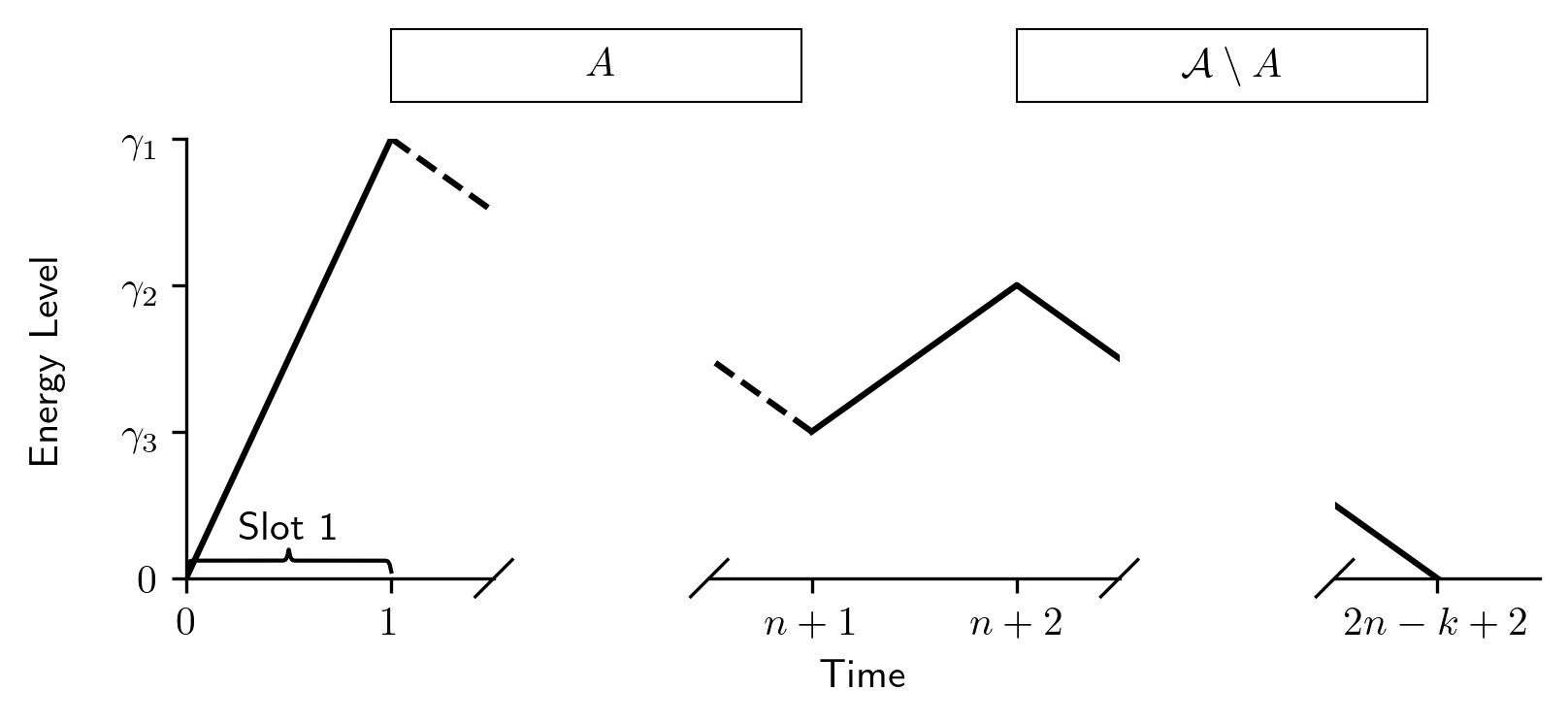}
        \caption{Feasible schedule structure associated with the reduction instance of \eas{} when jobs have arbitrary release times and identical due dates.}
        \label{fig:ksum_reduction_arb_rel}
    \end{figure}

    We now show that the total energy requirement of the $k$ jobs scheduled in $\tau_1$ is exactly $\gamma_1 =k \lrp{S^2 n^2} + \beta\lrp{S n}$.
    Assume for the sake of contradiction that this is not the case; that is, the total energy requirement of the $k$ jobs is $R = k\lrp{S^2n^2} + \beta^{\prime}\lrp{Sn}$ where $\beta^{\prime} \neq \beta$.
    If $\beta^{\prime} > \beta$, then by integrality it must be that $\beta^{\prime} \geq \beta + 1$, which implies $R \geq \gamma_1 + Sn$. However, as shown before, the total amount of energy
    harvestable on $\tau_1$ is $S(n-1)$, so it is not possible to feasibly schedule the $k$ jobs on $\tau_1$. Now consider the case when $\beta^{\prime} < \beta$.
    The energy harvested on $\tau_1$ is $S(n-k) - (S - \beta^{\prime})$ and $\lrp{\beta - \beta^{\prime}}\lrp{Sn}$ energy is leftover, so at the start of time slot $n+2$,
    there is at most $\gamma_2 = \lrp{\beta - \beta^{\prime}}\lrp{S n}+S(n-k-1)+ \beta^{\prime}$ energy. Since we must harvest energy at slot $n+2$, we get that at the start of slot $n+3$, the energy available is
    \begin{align*}
        \gamma_3 &= (n-k)(S^2 n^2) + (S-\beta)\lrp{S n} - S(n-k-1) - \beta + \gamma_2\\
                 &= (n-k)(S^2 n^2) + (S-\beta^{\prime})\lrp{Sn} - \beta+\beta^{\prime}.
    \end{align*}
    However, the energy requirement of the remaining $n-k$ jobs is $(n-k)(S^2 n^2) + (S-\beta^{\prime})\lrp{Sn}$, so it is not possible to feasibly schedule the
    remaining jobs in $\tau_2$. Hence, there is a contradiction so $\beta^{\prime} = \beta$ and $R = \gamma_1$.

    Therefore, a ``yes'' instance of this problem implies that the corresponding $\ksum$ instance is also a ``yes'' instance since the first $k$ jobs scheduled correspond to $k$ positive integers that sum to $\beta$.
    Conversely, a ``yes'' instance of $\ksum$ implies that the corresponding \eas{} instance is also a ``yes'' instance since we can schedule the jobs that correspond to the elements
    of $A$ at time slots in $\tau_1$ at their release time and the remaining $n-k$ jobs at time slots in $\tau_2$.
    Thus, \eas{} when jobs have arbitrary release times and identical due dates is $\NPH$.
\end{proof}

A similar type of reduction from \ksum{} can also be used to show the hardness of \eas{} with identical release times
and arbitrary due dates.

\begin{theorem}
    \label{thm:nph_arb_due}
    \eas{} when jobs have identical release times and arbitrary due dates is (weakly) $\NPH$.
\end{theorem}
\begin{proof}
    Given an instance of $\ksum$ as defined in the proof of \cref{thm:nph_arb_release}, construct a corresponding instance of \eas{} with the following:
    \begin{itemize}
        \item
        Time slots $\irange{1}{n+k+4}$ and a threshold value of $n+1$ jobs
        \item
        $\mathcal{J} = \lrc{J_1, \ldots, J_{n}, J_{n+1}}$. Each job $J_i \in \mathcal{J} \sm \lrc{J_{n+1}}$ has a release time, due date, and energy requirement of $r_i = 2, d_i = k+2+i,$ and $e_i = S^2 n^2+\alpha_i\lrp{S n}$, respectively.
        Job $J_{n+1}$ has $r_{n+1} = 2, d_{n+1} = n+k+4,$ and $e_{n+1} = S^3 n^3$.
        \item
        An energy harvesting profile defined by
        \[
        h_t = \begin{cases}
                    k\lrp{S^2n^2} + \beta\lrp{Sn} & \text{if } t = 1\\
                    0 & \text{if } t \in \irange{2}{k+1}\\
                    (n-k)\lrp{S^2 n^2} + (S-\beta)\lrp{Sn} & \text{if } t = k+2\\
                    S + \alpha_{t-k-2} & \text{if } t \in \irange{k+3}{n+k+2}\\
                    \lrp{S^3n^3} - Sk - \beta & \text{if } t = n+k+3\\
                    0 & \text{if } t = n+k+4
                \end{cases}
        \]
    \end{itemize}

    We claim that any feasible schedule that schedules all $n+1$ jobs must schedule exactly $k$ jobs from $\mathcal{J} \setminus \lrc{J_{n+1}}$ within time slots in $\tau_1 \coloneqq \irange{2}{k+1}$,
    the remaining $n-k$ jobs in $\mathcal{J} \setminus \lrc{J_{n+1}}$ within time slots in $\tau_2 \coloneqq \irange{k+3}{n+k+2}$, and job $J_{n+1}$ at time slot $n+k+4$.
    Additionally, by the construction of the energy profile, it is optimal to schedule jobs in $\tau_2$ at their due date as the earlier a job is scheduled, the more energy is lost from not harvesting energy at that time slot.
    We proceed with the assumption that jobs in $\tau_2$ are scheduled in this way.

    We first show the claim for job $J_{n+1}$. Note that it cannot be scheduled in $\irange{1}{n+k+2}$. This is because the total amount of energy harvestable on these time slots is
    \begin{equation*}
        n(S^2n^2) + S(Sn) + S(n+1) = S^2n(n^2 + 1) + S(n+1) < S^3n^3 + S(n+1).
    \end{equation*}
    If $J_{n+1}$ is scheduled on any of these slots, then the remaining energy is at most $S(n+1)$, which is strictly less than the energy required to
    schedule any other job. Hence, scheduling $J_{n+1}$ this early makes it impossible to schedule the remaining $n$ jobs. Additionally, note that $J_{n+1}$
    is the only job with a due date greater than $n+k+3$, so nothing is lost by harvesting energy at this slot and scheduling $J_{n+1}$ at slot $n+k+4$.

    We now show the claims for $\tau_1$ and $\tau_2$. We first note that jobs cannot be scheduled at time slots $1$ and $k+2$. Trivially, energy must be harvested at slot $1$ as no job is released yet. Now consider if a
    job was scheduled at slot $k+2$. Then the total energy harvestable on slots $\irange{1}{n+k+2}$ is at most $k(S^2n^2) + \beta{Sn} + S(n+1)$, which is strictly less
    than the energy required to schedule all the $n$ jobs in $\mathcal{J}\setminus\lrc{J_{n+1}}$ by their due dates. Thus, energy must be harvested at slot $k+2$.
    Now consider $\tau_1$. Clearly, no more than $k$ jobs can be scheduled on it. Assume for the sake of contradiction that $k^{\prime} < k$ jobs are scheduled on it instead.
    Because of the due date constraints, this requires $n-k^{\prime}$ jobs to be scheduled in $\tau_2$. In total, the energy requirement of the $n$ jobs scheduled in $\tau_1$ and $\tau_2$ is
    $n(S^2n^2) + S(Sn)$, which is exactly the amount harvested at time slots $1$ and $k+2$.
    Thus, at the start of time slot $n+k+3$, the only energy available is from the energy harvested at the open time slots of $\irange{k+3}{n+k+2}$, which is at most
    $Sk^{\prime} + S$. Since energy is harvested at slot $n+k+3$, the resulting energy at the start of slot $n+k+4$ is strictly less than
    $S^3n^3 - S(k - k^{\prime} - 1) - \beta < S^3n^3$. Hence, it is not possible to schedule job $J_{n+1}$, which is a contradiction. Therefore, exactly $k$ jobs in $\mathcal{J} \setminus \lrc{J_{n+1}}$ must be scheduled in $\tau_1$
    and the remaining $n-k$ jobs of $\mathcal{J} \setminus \lrc{J_{n+1}}$ in $\tau_2$.

    Next, we show that the energy requirement of the $k$ jobs scheduled in $\tau_1$ is exactly $\gamma_1 =k \lrp{S^2 n^2} + \beta\lrp{S n}$. Clearly, we do not have enough energy to schedule $k$ jobs that require more than $\gamma_1$ energy.
    Assume for the sake of contradiction that the energy requirement of these jobs is $k \lrp{S^2 n^2} + \beta^{\prime}\lrp{S n}$ where $\beta^{\prime} < \beta$.
    It follows that the energy harvested in $\tau_2$ is $Sk + \beta^{\prime}$ and thus the energy available at the start of time slot $n+k+4$ is $\gamma_2 = S^3n^3 - \beta + \beta^{\prime}$. To schedule $J_{n+1}$, we require
    $\gamma_2 \geq S^3n^3$, which gives $\beta^{\prime} \geq \beta$. Hence, there is a contradiction, so it must be that $\beta^{\prime} = \beta$.

    Therefore, a ``yes'' instance of this problem implies that the corresponding $\ksum$ instance is also a ``yes'' instance since the first $k$ jobs scheduled correspond to $k$ positive integers that sum to $\beta$.
    Conversely, a ``yes'' instance of $\ksum$ implies that the corresponding \eas{} instance is also a ``yes'' instance since we can schedule the $k$ jobs that correspond to the elements
    of $A$ at time slots in $\tau_1$, the $n-k$ jobs that correspond to elements of $\mathcal{A}\setminus{A}$ at their due dates in $\tau_2$, and $J_{n+1}$ at slot $n+k+4$.
    Thus, \eas{} when jobs have identical release times and arbitrary due dates is $\NPH$.
\end{proof}

Since jobs having arbitrary release times and identical due dates (or identical release times and arbitrary due dates) is a special case of them having both arbitrary release times and due dates,
we get the following as an immediate consequence of \cref{thm:nph_arb_release} (or \cref{thm:nph_arb_due}).

\begin{theorem}
    \eas{} when jobs have arbitrary release times and due dates is (weakly) $\NPH$.
    \label{cor:nph_arb}
\end{theorem}

\subsection{Weighted setting}
\label{subsec:nph_weighted}
In \weas{} it can be shown that the problem is $\NPH$ even when all the jobs have identical release time and due dates
through a straightforward reduction from \knapsack{}.

\begin{theorem}
    \weas{} when jobs have identical release times and due dates is (weakly) $\NPH$.
    \label{thm:nph_weighted}
\end{theorem}
\begin{proof}
    A reduction is given from \knapsack{}. An instance of this problem consists of a set of $n$ items where item $i$ has size $b_i$ and value $v_i$, a capacity $B$, and a value $V$.
    The objective is to decide whether there exists a subset $S \subseteq \lrc{1, \ldots, n}$ such that
    $\sum_{i \in S} b_i \leq B$ and $\sum_{i \in S} v_i \geq V$. Given an instance of \knapsack{}, we construct a corresponding instance
    of \weas{} with the following:
    \begin{itemize}
        \item Time slots $\irange{1}{n + 1}$ and a threshold value of $V$
        \item $\mathcal{J} = \lrc{J_1, \ldots, J_n}$, where each job $J_i \in \mathcal{J}$ has a release time, due date, energy requirement, and weight of $r_i = 2, d_i = n + 1, e_i = b_i$, and $w_i = v_i$, respectively
        \item An energy harvesting profile described by
        \[
            h_t = \begin{cases}
                        B & \text{if } t = 1\\
                        0 & \text{otherwise}
                    \end{cases}
        \]
    \end{itemize}

    Clearly, there is a ``yes'' instance of \knapsack{} if and only if the corresponding \weas{} instance is also a ``yes'' instance.
    Thus, \weas{} with identical release times and due dates is $\NPH$.
\end{proof}

Since the case of jobs having identical release times and due dates is a special case of either one or both of them being arbitrary, we get the following as an immediate consequence.

\begin{theorem}
    \weas{} is (weakly) $\NPH$.
    \label{cor:nph_weighted_arb}
\end{theorem}

\section{Conclusions and Open Problems}
\label{sec:conclusion}
We conclude with a brief summary of our results and open problems.
We presented three algorithms: (1) an optimal polynomial time algorithm for \eas{} with identical release times and due dates (\cref{sec:optimalA}),
(2) a greedy $\frac{1}{2}$-approximation algorithm for \eas{} with arbitrary release times and due dates (\cref{sec:greedy}), and (3) an FPTAS for \weas{} in the case of identical release times and due dates (\cref{sec:weighted_fptas}).

It would be interesting to see if there exists a PTAS or a better constant factor approximation for \eas{} with arbitrary release times and due dates, or if special cases
of \eas{} where only one of them is arbitrary admit better approximation ratios. Another natural direction
to consider is the extension of the greedy approach to \weas{} with arbitrary release times and due dates.

In \cref{sec:nph}, we study the hardness of both \eas{} and \weas{} and give nontrivial reductions from the $\ksum$ problem
to show that except for the case of identical release times and due dates, \eas{} is weakly $\NPH$ (\cref{subsec:nph_unweighted}). 
It is open whether \eas{} admits an FPTAS or whether there is a reduction from a strongly $\NPH$ problem to \eas{}.

One could also consider expanding our model. A natural extension is to consider the case of arbitrary processing times for jobs.
Another is to consider online versions of our problems. This may include either considering an online energy harvesting profile as considered in~\cite{IN2020},
or considering a model in which both jobs and the harvesting profile are revealed in an online manner.

\bibliographystyle{plain}
\bibliography{main}

\end{document}